\documentclass[pra,aps,floatfix,twocolumn,superscriptaddress,amsmath,amssymb]{revtex4-2}
\usepackage{graphicx}
\usepackage{hyperref}
\hypersetup{colorlinks=true,citecolor=blue,linkcolor=blue,urlcolor=blue}
\usepackage{bm}
\usepackage{braket}
\allowdisplaybreaks

\usepackage{amsthm}
\theoremstyle{definition}
\newtheorem{theorem}{Theorem}
\newtheorem{lemma}[theorem]{Lemma}

\begin{document}

\title{Stabilizer R\'enyi entropy of 3-uniform hypergraph states}

\author{Daichi Kagamihara}
\email{dkagamihara119@g.chuo-u.ac.jp}
\affiliation{%
Department of Physics, Chuo University, Bunkyo, Tokyo 112-8551, Japan}

\author{Shunji Tsuchiya}
\affiliation{%
Department of Physics, Chuo University, Bunkyo, Tokyo 112-8551, Japan}

\date{\today}

\begin{abstract}
Nonstabilizerness, also known as magic, plays a central role in universal quantum computation. Hypergraph states are nonstabilizer generalizations of graph states and constitute a key class of quantum states in various areas of quantum physics, such as the demonstration of quantum advantage, measurement-based quantum computation, and the study of topological phases. In this work, we investigate nonstabilizerness of 3-uniform hypergraph states, which are solely generated by controlled-controlled-Z gates, in terms of the stabilizer R\'enyi entropy (SRE). We find that the SRE of 3-uniform hypergraph states can be expressed using the matrix rank, which enables us to evaluate the quantity in $\mathcal{O}(N^3 2^{N})$ computational cost for $N$-qubit states. Furthermore, combined with Monte Carlo sampling, our findings provide an efficient method for estimating SRE. Based on this result, we exactly evaluate SREs of one-dimensional hypergraph states. We also present numerical results of SREs of several large-scale 3-uniform hypergraph states. Our results would contribute to an understanding of the role of nonstabilizerness in a wide range of physical settings where hypergraph states are employed.
\end{abstract}

\maketitle


\section{Introduction}
\label{sec_introduction}

Quantum computational advantage fundamentally relies on nonstabilizerness, also called magic, which enables quantum computations to outperform classical ones~\cite{Bravyi2005, Veitch2014, Howard2017}. The Gottesman-Knill theorem~\cite{gottesman1998,Aaronson2004} establishes that Clifford circuits acting on stabilizer states can be efficiently simulated by classical computers, highlighting the need for nonstabilizer resources to achieve universal quantum computation. To characterize such resources, several quantitative measures of nonstabilizerness have been proposed~\cite{Veitch2014,Bravyi2016,Howard2017,Bravyi2019,Wang2020,Leone2022,Warmuz2025,Tarabunga2025}. However, many of them involve high-dimensional optimizations over the stabilizer polytope, making them computationally intractable for large systems~\cite{Howard2017,Heinrich2019,Hamaguchi2024}. Recently, an optimization-free measure called the stabilizer R\'enyi entropy (SRE)~\cite{Leone2022,Leone2024} has gained attention as a practical tool for quantifying nonstabilizerness in multi-qubit systems, providing new insights into the structure of nonstabilizerness in complex quantum states~\cite{Lami2023,Tarabunga-Castelnovo2024,Turkeshi2025,Ding2025,hoshino2025,liu2025,Sinibaldi2025}.

Hypergraph states~\cite{Rossi2013,QuRi2013,salem2025} are generalizations of graph states~\cite{Hein2004,Hein2006} by allowing multi-qubit entangling operations, and have emerged as a key class of quantum states for demonstrating quantum advantage~\cite{Bremner2016}, measurement-based quantum computation (MBQC)~\cite{Miller2016,Takeuchi2019,Gachechiladze2019}, and topological phases~\cite{Yoshida2016,Miller2018}. A graph state is a multi-qubit quantum state defined by a graph, where each vertex represents a qubit initialized in $\ket{+} = (\ket{0} + \ket{1}) / \sqrt{2}$, and each edge represents a controlled-Z (CZ) gate applied between the connected qubits. Graph states are examples of stabilizer states. A hypergraph state generalizes this construction by introducing hyperedges, which correspond to multi-qubit controlled-Z gates acting simultaneously on three or more qubits. Since such multi-qubit controlled-Z gates (e.g., controlled-controlled-Z (CCZ) gate) are non-Clifford operations, the resulting hypergraph states are generally nonstabilizer states and thus serve as useful resources for quantum computation beyond the stabilizer formalism.

Both graph states and hypergraph states are typical resource states in MBQC, but they are used in quite different ways. MBQC is the quantum computation model that consists of two elements: preparing an entangled resource state and performing quantum computation by adaptively measuring each qubit in a suitable basis. MBQC using a cluster state, a type of graph state, needs arbitrary single-qubit measurements~\cite{Raussendorf2003}. They supply nonstabilizerness to the state, and then may enable universal computations despite the lack of nonstabilizerness of the resource states~\cite{Gong-Chu2025}. In contrast, MBQC with a hypergraph state can achieve universal computations by using only Pauli basis measurements~\cite{Miller2016,Takeuchi2019,Gachechiladze2019}, which is called the Pauli universality. Since the Pauli measurements do not provide any nonstabilizerness, the nonstabilizerness of the resource state may play a crucial role in hypergraph state MBQC. This raises the question of which characteristics of nonstabilizerness of hypergraph states are relevant to the Pauli universality.

In this paper, we investigate the nonstabilizerness of 3-uniform hypergraph states, some of which are used in the hypergraph state MBQC. A hypergraph state generated exclusively by CCZ gates is referred to as a 3-uniform hypergraph state~\cite{Rossi2013}. We adopt SRE as a computable measure of nonstabilizerness. We show that the SRE of 3-uniform hypergraph states can be expressed using the matrix rank. This expression results in an exponential reduction in the computational cost of the SRE of these states. Utilizing this advantage, we numerically evaluate the SREs of large-scale 3-uniform hypergraph states.

This paper is organized as follows: In Sec.~\ref{sec:preliminaries}, we introduce the definitions of SRE and 3-uniform hypergraph states. We also review previous work on the SRE of hypergraph states~\cite{Chen2024}. In Sec.~\ref{sec:results}, we first show our main result, Theorem~\ref{theorem:main}. Then we discuss physical interpretations and consequences of our results. We also present exact and numerical evaluations of some hypergraph states. In Sec.~\ref{sec:summary}, we summarize our results and discuss future directions.

\section{Preliminaries}
\label{sec:preliminaries}

Throughout this paper, we consider $N$-qubit systems. Stabilizer R\'enyi entropy (SRE) $M_{\alpha}$ is defined as~\cite{Leone2022,Leone2024,Chen2024}
\begin{align}
M_{\alpha} &= \frac{1}{1-\alpha} \log_{2} m_{\alpha},
\label{eq:def_SRE}
\\
m_{\alpha} &= \frac{1}{2^{N}} \sum_{\bm{x},\bm{z}} | \bra{\psi} P^{\bm{x},\bm{z}} \ket{\psi}|^{2\alpha},
\end{align}
where $P^{\bm{x},\bm{z}}$ is a Pauli string, which is parametrized by two bit strings $\bm{x} = (x_0, x_1, \dots, x_{N-1})$ and $\bm{z}$: $P^{\bm{x},\bm{z}} = \sqrt{-1}^{\bm{x} \cdot \bm{z}} X^{\bm{x}} Z^{\bm{z}}$, where we denote $X^{\bm{x}} = X_0^{x_0} X_1^{x_1} \cdots X_{N-1}^{x_{N-1}}$. $m_{\alpha}$ is called the Pauli-Liouville (PL) moment~\cite{Chen2024}. SRE satisfies the resource theory requirements~\cite{Leone2022}: (i) Invariance under Clifford operations $C$; $M_{\alpha}(C \ket{\psi}) = M_{\alpha}(\ket{\psi})$, and (ii) Faithfulness; $M_{\alpha}(\ket{\psi}) = 0$ iff $\ket{\psi}$ is a pure stabilizer state, and otherwise $M_{\alpha}(\ket{\psi}) > 0$. These ensure that SRE is a good measure of nonstabilizerness. In addition, for $\alpha \geq 2$, $M_{\alpha}$ is a stabilizer monotone under stabilizer protocols~\cite{Leone2024}.

The 3-uniform hypergraph state is described as~\cite{Rossi2013,Chen2024}
\begin{align}
\ket{\psi} = \prod_{(i,j,k) \in E_3} CCZ_{i,j,k} \ket{+}^{\otimes N},
\label{eq:def_3-uniform_hypergraph_state}
\end{align}
where $E_{3}$ is a set of three-qubit hyperedges. The CCZ gate is defined as
\begin{align}
CCZ_{i,j,k} = I_{i,j,k} - 2\ket{1,1,1}_{i,j,k} \bra{1,1,1}_{i,j,k},
\end{align}
where $I_{i,j,k}$ is the three-qubit identity operator and $i$, $j$, and $k$ denote the qubit indices. By defining $U \equiv \prod_{(i,j,k) \in E_3} CCZ_{i,j,k}$, one can express the corresponding density matrix in the generalized-stabilizer form~\cite{Chen2024}:
\begin{gather}
\rho = |\psi \rangle \langle \psi| = \frac{1}{2^{N}} \sum_{\bm{s}} S^{\bm{s}},~S^{\bm{s}} = g_{0}^{s_0} g_1^{s_1} \dots g_{N-1}^{s_{N-1}},
\\
g_i = U X_i U,~g_{i} \ket{\psi} = \ket{\psi},~(i=0,1,\dots,N-1).
\end{gather}
Note that the generalized generator $g_{i}$ is not a Pauli string and has the form
\begin{align}
g_i = X_i \prod_{j,k: (i,j,k) \in E_3} CZ_{j,k},
\end{align}
where we used $CCZ_{i,j,k} X_i CCZ_{i,j,k} = X_{i} CZ_{j,k}$.

The previous work~\cite{Chen2024} investigated the SRE of general hypergraph states (not restricted to 3-uniform ones) and obtained useful relations. We present some of their results relevant to our study, adapted to the current context. The generalized stabilizer $S^{\bm{s}}$ can be expressed in a simple form:
\begin{align} \notag
S^{\bm{s}} &= (-1)^{\sum_{(i,j,k) \in E_3} s_i s_j s_k} X^{\bm{s}}
\\
&\times \prod_{i=0}^{N-1} \left( Z_{i}^{\sum_{j,k: (i,j,k) \in E_3} s_j s_k} \prod_{j,k: (i,j,k) \in E_3} CZ_{j,k}^{s_i} \right).
\label{eq:stabilizer}
\end{align}
The expectation value of a Pauli string with respect to a 3-uniform hypergraph state is given by
\begin{align}
|\bra{\psi} P^{\bm{x},\bm{z}} \ket{\psi}|
= \left| \frac{1}{2^N} \sum_{\bm{s}} \mathrm{Tr} (P^{\bm{x},\bm{z}} S^{\bm{s}}) \right|.
\end{align}
From Eq.~\eqref{eq:stabilizer}, $\mathrm{Tr} (P^{\bm{x},\bm{z}} S^{\bm{s}})$ has a finite value only when $\bm{s} = \bm{x}$ because Z and CZ operators give only phases to the computational basis. Then, we can express it as
\begin{align} \notag
&|\bra{\psi} P^{\bm{x},\bm{z}} \ket{\psi}|
= \left| \frac{1}{2^N} \sum_{\bm{a}} \bra{\bm{a}} P^{\bm{x},\bm{z}} S^{\bm{x}} \ket{\bm{a}} \right|
\\ \notag
&= \Big| \frac{1}{2^N} 
\sum_{\bm{a}} (-1)^{\sum_{i=0}^{N-1} a_i (z_i + \sum_{j,k: (i,j,k) \in E_3} x_j x_k)}
\\
&\quad \quad \times (-1)^{\sum_{i, j,k: (i,j,k) \in E_3} a_i a_j x_k} \Big|,
\end{align}
where $\bm{a}$ is an $N$-bit string, $\ket{\bm{a}} = \ket{a_0}_0 \otimes \ket{a_1}_1 \otimes \dots \otimes \ket{a_{N-1}}_{N-1}$ represents the computational basis. By rearranging $\bm{z}$, we obtain the PL moment of a 3-uniform hypergraph state as
\begin{align}
m_{\alpha}
=
\frac{1}{2^{N}} \sum_{\bm{x},\bm{z}} \left[ \frac{1}{2^{N}} \sum_{\bm{a}} (-1)^{\bm{z} \cdot \bm{a} + \sum_{i<j} C_{i,j}(\bm{x}) a_i a_j} \right]^{2\alpha},
\label{eq:m_prev_work}
\end{align}
where we define $C_{i,j}(\bm{x}) = \sum_{k: (i,j,k) \in E_3} x_k$. $C_{i,j}$ is an upper-triangular matrix with zero diagonal entries, and the summation of $k$ in $C_{i,j}$ is taken modulo 2.

\section{Results}
\label{sec:results}

In this section, we first show our main result and its proof in Sec.~\ref{sec:Result_main}.
In Sec.~\ref{sec:Result_PhysMean}, we discuss the physical meaning of our result.
In Sec.~\ref{sec:Result_upper_bound}, we derive an upper bound for SRE, as a corollary of our results.
Then we show the evaluations of SREs of 3-uniform hypergraph states.
In Sec.~\ref{sec:Result_SRE_1d}, we consider one-dimensional hypergraph states and show how to compute their SREs exactly.
We then consider the two-dimensional hypergraph states and estimate their SREs for large-scale hypergraph states in Sec.~\ref{sec:result_SRE_2d}.

\subsection{Main result}
\label{sec:Result_main}

Our main result is the following.
\begin{theorem}
The PL moment of a 3-uniform hypergraph state is given by
\begin{align}
m_{\alpha} = \frac{1}{2^N} \sum_{\bm{x}} 2^{(1-\alpha) 2h(\bm{x})},
\label{eq:m_present_work}
\end{align}
where $2h(\bm{x})$ (always even) is the rank of the GF(2) symmetric matrix $C(\bm{x}) + C(\bm{x})^T$. 
\label{theorem:main}
\end{theorem}

\begin{proof}
The exponent of $(-1)$ in Eq.~\eqref{eq:m_prev_work}, $\bm{z} \cdot \bm{a} + \bm{a}^T C(\bm{x}) \bm{a}$, can be treated over GF(2). Here, GF(2) denotes the finite field with elements \{0,1\}, where addition corresponds to the exclusive-or (XOR) operation and multiplication corresponds to the logical AND operation. From the theory of the finite field quadratic form~\cite{Lidl_Niederreiter_1996}~(for completeness, we review them in Appendix~\ref{appendix_A}), the quadratic form $Q = \bm{a}^T C(\bm{x}) \bm{a}$ can be transformed into a standard form by a linear transformation $\bm{a} = P\bm{a}'$, where $P$ is a non-singular matrix over GF(2) (so that the transformation is composed of only XOR operations):
\begin{align}
Q' =
\begin{cases}
a'_0 a'_1 + \dots +  a'_{r-3} a'_{r-2} + a'_{r-1} & (\text{$r$: odd}),
\\
a'_0 a'_1 + \dots +  a'_{r-2} a'_{r-1} + \eta (a'_{r-2} + a'_{r-1}) & (\text{$r$: even}),
\end{cases}
\label{eq:GF2_standard_form}
\end{align}
where $r$ represents the dimension of the non-degenerate subspace of $Q$, and $\eta = 0$ or $1$. As shown in Appendix~\ref{appendix_B}, the rank $2h(\bm{x})$ of the GF(2) symmetric matrix $C(\bm{x})+ C(\bm{x})^T$ corresponds to twice the number of quadratic pairs, that is, $r-1 ~ (r)$ when $r$ is odd (even). Using this transformation, Eq.~\eqref{eq:m_prev_work} can be decoupled as, for odd $r[=2h(\bm{x})+1]$,
\begin{align} \notag
&\sum_{\bm{a}} (-1)^{\bm{z} \cdot \bm{a}+ Q}
\\ \notag
&=
\sum_{\bm{a}'} (-1)^{\bm{z}' \cdot \bm{a}'} \left[ \prod_{j=0}^{h(\bm{x})-1} (-1)^{a'_{2j} a'_{2j+1}} \right] (-1)^{a'_{2h(\bm{x})}} 
\\ \notag
&=
\prod_{j=0}^{h(\bm{x})-1} \left[ \sum_{a'_{2j}, a'_{2j+1}} (-1)^{z'_{2j} a'_{2j} + z'_{2j+1} a'_{2j+1}} (-1)^{a'_{2j} a'_{2j+1}} \right] 
\\ \notag
&\times
\sum_{a'_{2h(\bm{x})}} (-1)^{z'_{2h(\bm{x})} a'_{2h(\bm{x})}} (-1)^{a'_{2h(\bm{x})}} \times \prod_{k=2h(\bm{x})+1}^{N-1} \sum_{a'_k} (-1)^{z'_k a'_k},
\end{align}
where we define $\bm{z}' = P^T \bm{z}$. By taking $\bm{a}'$ summations explicitly, we obtain 
\begin{align} \notag
\frac{1}{2^{N}} \sum_{\bm{a}} (-1)^{\bm{z} \cdot \bm{a}+ Q}
&=
2^{-h(\bm{x})} \prod_{j=0}^{h(\bm{x})-1} (-1)^{z'_{2j} z'_{2j+1}}
\\ \notag
&\quad \times \delta_{z'_{2h(\bm{x})},1} \times \prod_{k=2h(\bm{x})+1}^{N-1} \delta_{z'_k,0}.
\end{align}
In the same way, for even $r[=2h(\bm{x})]$,
\begin{align} \notag
&\frac{1}{2^{N}} \sum_{\bm{a}} (-1)^{\bm{z} \cdot \bm{a}+ Q}
=
2^{-h(\bm{x})} \prod_{j=0}^{h(\bm{x})-2} (-1)^{z'_{2j} z'_{2j+1}}
\\ \notag
&\quad
\times  (-1)^{(z'_{2h(\bm{x})-2} +\eta) (z'_{2h(\bm{x})-1} + \eta)} \times \prod_{k=2h(\bm{x})}^{N-1} \delta_{z'_k,0}.
\end{align}
In both cases, the absolute values of these terms are given by $2^{-h(\bm{x})}$ when $z'_{k}$ is zero for $k = 2h(\bm{x}), 2h(\bm{x})+1, \dots, N-1$, and $0$ when at least one $z'_k$ is unity for $k=2h(\bm{x}), 2h(\bm{x})+1,\dots, N-1$. The number of non-zero terms is $2^{2h(\bm{x})}$, corresponding to the pattern of $z'_0, \dots, z'_{2h(\bm{x})-1}$. Thus, we also analytically take the $\bm{z}~(\bm{z}')$ summation:
\begin{align} \notag
\sum_{\bm{z}} \left|\frac{1}{2^{N}} \sum_{\bm{a}} (-1)^{\bm{z} \cdot \bm{a}+ Q} \right|^{2\alpha} = 2^{2h(\bm{x})} \times \left[ 2^{-h(\bm{x})} \right]^{2\alpha}.
\end{align}
Substituting this into Eq.~\eqref{eq:m_prev_work}, we obtain Eq.~\eqref{eq:m_present_work}.
\end{proof}

Once $m_{\alpha}$ is determined, SRE is obtained from Eq.~(\ref{eq:def_SRE}). The cases $\alpha = 1$ and $\alpha = \infty$ should be treated separately; $M_{\alpha \to 1} = \sum_{\bm{x}} 2h(\bm{x}) / 2^{N}$ and $M_{\alpha \to \infty} = 0$. The former is simply obtained by taking the limit $\alpha \to 1$. The latter is formally expressed as $M_{\alpha \to \infty} = \min_{\bm{x}} 2h(\bm{x})$ and the minimization is achieved for $\bm{x} = \bm{0}$ and $h(\bm{x} = \bm{0}) = 0$ since $C(\bm{x}=\bm{0})$ is the zero matrix.

The computational cost of evaluating the rank $2h(\bm{x})$ is $\mathcal{O}(N^3)$, so that the $\bm{a}$ and $\bm{z}$ summations, which require an $\mathcal{O}(4^{N})$ cost when evaluated according to the definition, are replaced by a polynomial time procedure. Thus, our result allows the evaluation of SREs in $\mathcal{O}(N^32^{N})$, which is exponentially faster than the naive approach from the definition $\mathcal{O}(2^{3N})$ and the recent improved algorithm $\mathcal{O}(N2^{2N})$~\cite{huang2026,xiao2026,Sierant2026}. Our result still has the $\bm{x}$ summation, and thus it requires an exponential cost for an exact evaluation. For an approximate evaluation, the cost of the $\bm{x}$ summation can be reduced by using a Monte Carlo technique. Note that, as seen from Eq.~\eqref{eq:m_present_work}, the Monte Carlo sampling is free from the sign problem, and the quantity can therefore be estimated efficiently in polynomial time.

To intuitively understand why our result allows faster evaluation, it is helpful to rewrite the expectation value of a Pauli string as
\begin{align}
\left| \bra{\psi} P^{\bm{x},\bm{z}} \ket{\psi} \right| = \left| \bra{+}^{\otimes N} Z^{\bm{z}} S^{\bm{x}} \ket{+}^{\otimes N} \right|.
\end{align}
Here, $Z^{\bm{z}} S^{\bm{x}} \ket{+}^{\otimes N}$ and $\ket{+}^{\otimes N}$ are stabilizer states since the generalized stabilizer $S^{\bm{x}}$ is composed of X, Z and CZ operators. Therefore, their inner product can be evaluated in an $\mathcal{O}(N^3)$ cost by utilizing the Aaronson-Gottesman tableau algorithm~\cite{Aaronson2004}. The efficient evaluation of $|\bra{\psi} P^{\bm{x}, \bm{z}} \ket{\psi}|$ in our method would be attributed to the Aaronson-Gottesman algorithm, the generalization of the Gottesman-Knill theorem.

Before going to the next topic, we briefly comment on the inclusion of hyperedges containing one and two qubits, which corresponds to the inclusion of Z and CZ gates in Eq.~\eqref{eq:def_3-uniform_hypergraph_state}. Their presence does not affect our result Theorem~\ref{theorem:main}, because they are Clifford operations and cannot contribute to nonstabilizerness. Thus, our result can apply to hypergraph states with hyperedges containing at most three qubits.

\subsection{Physical meaning of our result}
\label{sec:Result_PhysMean}

To gain further insight, let us consider the physical interpretation of Theorem~\ref{theorem:main}. From the derivation, the $\bm{x}$ summation in Eq.~\eqref{eq:m_present_work} corresponds to the summation over generalized stabilizers. For a fixed $\bm{x}$, the product of CZ gates in a generalized stabilizer $S^{\bm{x}}$~\eqref{eq:stabilizer} can be rewritten as
\begin{align}
\prod_{i=0}^{N-1} \prod_{j,k: (i,j,k) \in E_3} CZ_{j,k}^{x_i} = \prod_{i<j} CZ_{i,j}^{C_{i,j}(\bm{x})}.
\end{align}
This implies that $C_{i,j}(\bm{x}) = 1(0)$ represents the presence (absence) of a CZ gate on $(i,j)$ qubits in the generalized stabilizer $S^{\bm{x}}$. Using the diagonal representation of a CZ gate $CZ_{i,j}$ in the computational basis,
\begin{align}
CZ_{i,j} = \sum_{\bm{a}} (-1)^{a_ia_j} \ket{\bm{a}} \bra{\bm{a}},
\end{align}
the product of CZ gates can be rewritten as
\begin{align}
\prod_{i<j} CZ_{i,j}^{C_{i,j}(\bm{x})} = \sum_{\bm{a}} (-1)^{\sum_{i<j} C_{i,j}(\bm{x}) a_ia_j} \ket{\bm{a}} \bra{\bm{a}}.
\label{eq:prod_CZ_expr_comp_basis}
\end{align}
The quadratic form $Q(\bm{a})=\sum_{i<j} C_{i,j}(\bm{x}) a_i a_j$ appears in the exponent of $(-1)$.

Next, we apply a GF(2) linear transformation $\bm{a} = P \bm{a}'$ used in Eq.~\eqref{eq:GF2_standard_form}. Let $U_{\rm CNOT}$ be a unitary operation which encodes the relation $\bm{a} = P \bm{a}'$ in the computations bases as $\ket{\bm{a}} = U_{\rm CNOT} \ket{\bm{a}'} (=\ket{P\bm{a}'})$. Because any GF(2) linear transformation can be expressed as a sequence of XOR operations, we can represent $U_{\rm CNOT}$ as the product of controlled-NOT (CNOT) operations. Substituting these into Eq.~\eqref{eq:prod_CZ_expr_comp_basis}, we obtain
\begin{align} \notag
&\prod_{i<j} CZ_{i,j}^{C_{i,j}(\bm{x})} = \sum_{\bm{a}} (-1)^{Q(P \bm{a}')} \ket{P \bm{a}'} \bra{P \bm{a}'}
\\
&\quad = U_{\rm CNOT} \left( \sum_{\bm{a}'} (-1)^{Q'(\bm{a}')} \ket{\bm{a}'} \bra{\bm{a}'} \right) U_{\rm CNOT}^{\dag}.
\end{align}
Since the quadratic form after the transformation $Q'(\bm{a}')$ has the form \eqref{eq:GF2_standard_form}, the above observation implies that there exists $U_{\rm CNOT}$ such that
\begin{align} \notag
&U^{\dagger}_{\rm CNOT} \left( \prod_{i<j} CZ_{i,j}^{C_{i,j}(\bm{x})} \right) U_{\rm CNOT}
\\
&=\begin{cases}
\displaystyle \prod_{i=0}^{h(\bm{x})-1} CZ_{2i,2i+1} Z_{2h(\bm{x})} & (r\text{:\,odd}),
\\
\displaystyle \prod_{i=0}^{h(\bm{x})-1} CZ_{2i,2i+1} Z_{2h(\bm{x})-2}^{\eta} Z_{2h(\bm{x})-1}^{\eta} & (r \text{:\,even}).
\end{cases}
\label{eq:simplification_of_product_CZ_by_CNOT}
\end{align}
In this way, the rank $2h(\bm{x})$ of a GF(2) symmetric matrix $C(\bm{x}) + C(\bm{x})^T$ in Theorem~\ref{theorem:main} is interpreted as twice the minimum number of CZ gates when decomposing a generalized stabilizer $S^{\bm{x}}$ using CNOT operations. This is a natural interpretation in the sense that the minimum number of non-Pauli operators included in a generalized stabilizer under Clifford transformations contributes to nonstabilizerness.

\subsection{Upper bound of SRE of 3-uniform hypergraph states}
\label{sec:Result_upper_bound}

As a simple corollary of Theorem~\ref{theorem:main}, we can obtain an upper bound of SRE for $\alpha > 1$. The rank $2h(A)$ of a GF(2) symmetric matrix $A$ satisfies the subadditivity: $2h(A+B) \leq 2h(A) + 2h(B)$. We decompose $C_{i,j}(\bm{x})$ as $C_{i,j}(\bm{x}) = \sum_{k: (i,j,k) \in E_3} x_k = \sum_{k=0}^{N-1} x_k \delta_{(i,j,k) \in E_{3}} = \sum_{k=0}^{N-1} x_k (C'_k)_{i,j}$, where $(C'_k)_{i,j} = \delta_{(i,j,k) \in E_{3}}$ and $\delta_{\text{condition}}$ is 1 if the condition is satisfied and 0 otherwise. The subadditivity implies for $\alpha > 1$
\begin{align} \notag
m_{\alpha} &\geq \frac{1}{2^{N}} \sum_{\bm{x}} 2^{-(\alpha-1)\sum_{k=0}^{N-1} x_k 2 h_k}
\\
&=
\frac{1}{2^{N}} \prod_{k=0}^{N-1} \left[ 1 + 2^{-(\alpha-1) 2 h_k} \right],
\end{align}
where $2h_k = 2h(C'_k+(C'_k)^T)$. Note that $C'_k$ is the same as $C(\bm{x})$ for $x_k=1$ and $x_i=0 ~(i \neq k)$. From the discussion in Sec.~\ref{sec:Result_PhysMean}, $h_k$ represents the minimum number of CZ gates under CNOT transformation in the generalized generator $g_k$, which corresponds to the generalized stabilizer $S^{\bm{x}}$ for $x_k = 1$ and $x_i = 0~(i\neq k)$. As a result, SRE is bounded from above:
\begin{align} \notag
M_{\alpha}
&\leq
\frac{1}{\alpha-1} \sum_{k=0}^{N-1} \left[ 1-\log_2\left( 1 + 2^{-(\alpha-1)2h_k} \right) \right]
\\
&\leq
\frac{N}{\alpha-1} \left[ 1 - \log_2\left( 1 + \frac{1}{2^{(\alpha-1)2\bar{h}}} \right) \right],
\label{eq:upper_bound_present_work}
\end{align}
where $\bar{h} = \sum_{k=0}^{N-1} h_k / N$ and we used Jensen's inequality.

The previous work~\cite{Chen2024} obtained a similar upper bound for general hypergraph states, for $\alpha > 1$,
\begin{align}
S_{\alpha} \leq \frac{N}{\alpha-1} \left[ 1- \log_2\left( 1 + \frac{1}{2^{(2\alpha-1)\bar{\Delta}}} \right) \right],
\label{eq:upper_boud_prev_work}
\end{align}
where $\bar{\Delta} = \sum_{k=0}^{N-1} \Delta_k / N$ and $\Delta_k$ is the number of hyperedges containing the $k$th vertex. For 3-uniform hypergraph states, $\Delta_k = \sum_{i<j} (C'_{k})_{i,j}$, which means that $\Delta_k$ is also understood as the number of CZ gates in the generalized generator $g_k$. By their definitions, $h_k \leq \Delta_k$ always holds. Thus, our upper bound is tighter than the previous work.

\begin{figure}[t]
\centering
\includegraphics[width=\columnwidth]{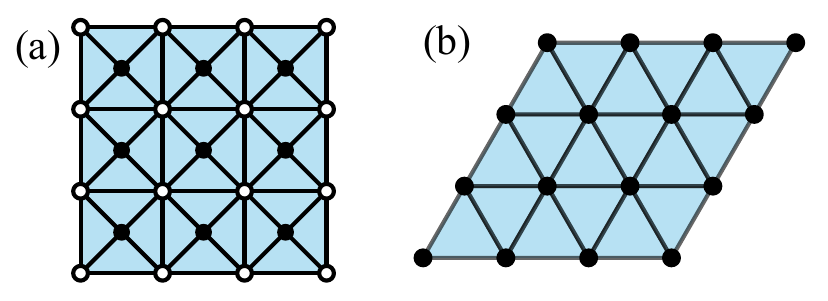}
\caption{
(a) Union Jack and (b) triangular lattice hypergraph states for the linear length $L=3$. Both filled and unfilled dots denote qubits. CCZ operations are applied to each of the smallest triangles.
}
\label{fig:lattice_unionjack_triangular}
\end{figure}

As examples of $\bar{h}$ and $\bar{\Delta}$, we take the $L \times L$ Union Jack lattice and triangular lattice hypergraph states, which are shown in Fig.~\ref{fig:lattice_unionjack_triangular}(a) and (b), respectively. For simplicity, we assume the periodic boundary condition. In the Union Jack lattice case, the rank $h_k$ corresponding to unfilled dots in Fig.~\ref{fig:lattice_unionjack_triangular}(a) is 3 and that to filled dots is 1, so that the average rank $\bar{h} = 2$ (For the $L=2$ case, $h_k = 1$ on unfilled dots, then $\bar{h}=1$). In the triangular lattice case, the rank $h_k = 2$ everywhere, therefore $\bar{h} = 2$. In both cases, $\bar{\Delta}=6$. For more detail, see Appendix~\ref{appendix:h_Delta_calc}.

\subsection{Exact evaluation of SREs of one-dimensional hypergraph states}
\label{sec:Result_SRE_1d}

Let us evaluate the SRE of concrete examples of 3-uniform hypergraph states. Even though our result~\eqref{eq:m_present_work} still contains the $\bm{x}$ summation with exponential computational cost, we can exactly calculate the PL moment of some one-dimensional hypergraph states. As discussed in Sec.~\ref{sec:Result_PhysMean}, the rank of $C(\bm{x}) + C(\bm{x})^T$ corresponds to twice the minimum number of CZ gates contained in a generalized stabilizer $S^{\bm{x}}$ when decomposing it by a CNOT unitary operation. As explained below, when decoupling $S^{\bm{x}}$ into one CZ gate and the remaining part using CNOT transformations, the remainder takes a similar form to that before the decoupling for some one-dimensional hypergraph states. In this situation, the number of relevant qubits in the remaining part is reduced. Utilizing this fact, one can construct a recursion relation for the PL moments of hypergraph states having different numbers of qubits.

\begin{figure}[t]
\centering
\includegraphics[width=\columnwidth]{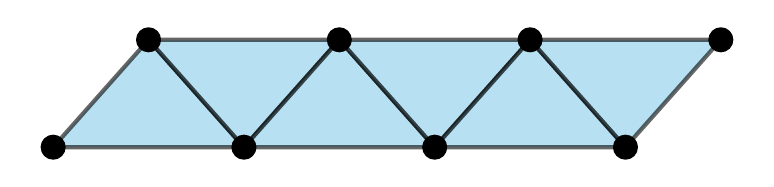}
\caption{
An example of the 3-uniform hypergraph state whose PL moment can be exactly calculated. We show the $N = 8$ case. Dots denote qubits and triangles represent CCZ operations.}
\label{fig:lattice_1d_triangular}
\end{figure}

For concreteness, we consider a hypergraph state with hyperedges
\begin{align}
E^{(N)}_3 = \{(0,1,2), (1,2,3), \cdots, (N-3,N-2,N-1)\},
\label{eq:EN_3_1d}
\end{align}
which is shown in Fig.~\ref{fig:lattice_1d_triangular}. Let us consider the decoupling of CZ gates in a generalized stabilizer $S^{\bm{x}}$ into a CZ gate on the $(0,1)$ qubit pair and the remainder. We denote the product of CZ gates in a generalized stabilizer $S^{\bm{x}}$ as $P_{\rm CZ}(S^{\bm{x}})$, which is given by
\begin{align} \notag
P_{\rm CZ}(S^{\bm{x}}) &= CZ^{x_2}_{0,1} CZ_{0,2}^{x_1} CZ_{1,2}^{x_0+x_3} CZ_{1,3}^{x_2} CZ_{2,3}^{x_1+x_4} 
\\
& \quad \times CZ_{2,4}^{x_3} CZ_{3,4}^{x_2} \times \prod_{i, j,k: (i,j,k) \in E'_3} CZ_{j,k}^{x_i},
\end{align}
where $E'_{3} = \{(3,4,5), (4,5,6),\dots, (N-3,N-2,N-1) \}$. A CZ gate on the $(0,1)$ qubits should already be contained in $P_{\rm CZ}(S^{\bm{x}})$, i.e., $x_2 = 1$ is necessary. We denote a CNOT (controlled-X) operator acting on a control qubit $k$ and a target qubit $l$ as $CX_{k,l}$, which is given by
\begin{align}
CX_{k,l} = \ket{0}\bra{0}_{k} \otimes I_l + \ket{1}\bra{1}_{k} \otimes X_{l}.
\end{align}
The conjugation by $U_{\rm CNOT} = CX_{2, 1}^{x_1} CX_{2, 0}^{x_0+x_3} CX_{3,0}$ decouples $P_{\rm CZ}(\left. S^{\bm{x}} \right|_{x_2 = 1})$ as
\begin{align} \notag
&U_{\rm CNOT}^{\dagger} P_{\rm CZ}(\left. S^{\bm{x}} \right|_{x_2=1}) U_{\rm CNOT} = CZ_{0,1} Z_2^{x_0x_1 + x_1x_3} 
\\ 
&\quad \times CZ_{2,3}^{x_4} CZ_{2,4}^{x_3} CZ_{3,4} \times \prod_{i, j,k: (i,j,k) \in E'_3} CZ_{j,k}^{x_i},
\label{eq:UCNOT_case_i}
\end{align}
where we used, with $a$ and $b$ being GF(2) variables,
\begin{align}
CX^{a}_{i,j} CZ^{b}_{i,j} CX^{a}_{i,j} &= Z^{ab}_{i} CZ_{i,j}^{b},
\\
CX^{a}_{i,j} CZ^{b}_{j,k} CX^{a}_{i,j} &= CZ^{ab}_{i,k} CZ_{j,k}^{b}.
\end{align}
Since the $Z_2^{x_0x_1 + x_1x_3}$ term does not contribute to the number of CZ gates and nonstabilizerness, we can ignore it. Looking at Eq.~\eqref{eq:UCNOT_case_i} other than $CZ_{0,1}$ and $Z_2^{x_0x_1 + x_1x_3}$, it has the same form as the product of CZ gates of a generalized stabilizer of a hypergraph state with hyperedges $E''_{3} = \{(2,3,4), (3,4,5),\dots, (N-3,N-2,N-1) \}$ with an additional condition $x_2 = 1$. Note that $E''_{3}$ does not contain the zeroth and first qubits, and therefore, $x_0$ and $x_1$ do not appear.

Based on this observation, the $N$-qubit PL moment can be related to the corresponding $(N-2)$-qubit one. We define the PL moment of the hypergraph state with hyperedges $E^{(N)}_3$ with fixed values $\{x_i = y_i\}_{i \in A}$ as
\begin{align}
m_{\alpha,(\{x_i = y_i\}_{i \in A})}^{(N)}= \frac{1}{2^{N-|A|}} \sum_{\bm{x}, x_i = y_i (i \in A)} 2^{(1-\alpha)2h(\bm{x})},
\end{align}
where $A$ is the set of indices of fixed $x$ values, and $|A|$ denotes the number of fixed values. The above discussion implies that $m_{\alpha,(x_2=1)}^{(N)}$ is equivalent to a single CZ gate contribution to the PL moment $2^{(1-\alpha)2}$ times the PL moment of a hypergraph state with hyperedges $E''_{3}$ with $x_2 = 1$. The latter takes the same value as the PL moment with hyperedges $E^{(N-2)}_{3}$ with $x_0=1$, because $E''_{3}$ does not contain the zeroth and first qubits. As a result, we obtain
\begin{align}
m_{\alpha,(x_2=1)}^{(N)} = \beta m_{\alpha,(x_0=1)}^{(N-2)},
\end{align}
where $\beta = 2^{(1-\alpha)2}$. Since this result does not depend on $x_i~(i \neq 2)$, the same relation holds even when we fix some variables other than $x_2$, such as, 
\begin{align}
m_{\alpha,(x_0=y_0,x_1=y_1,x_2=1,x_3=y_3)}^{(N)} = \beta m_{\alpha,(x_0=1,x_1=y_3)}^{(N-2)},
\end{align}
where $y_i~(i=0,1,3)$ are GF(2) variables.

\begin{table*}
\centering
\caption{Summary of the results of the decoupling of a CZ gate in a generalized stabilizer of the hypergraph state with hyperedges $E^{(N)}_3$ given by Eq.~\eqref{eq:EN_3_1d}. We can extract a CZ gate on (i) $(0,1)$, (ii) $(0,2)$, (iii) $(1,2)$ qubits, or (iv) we cannot extract a CZ gate on any pair of $(0,1,2)$ qubits. The condition column lists the condition required for each case. $U_{\rm CNOT}$ corresponds to the CNOT operations which decompose the product of CZ gates in a generalized stabilizer $P_{\rm CZ}(S^{\bm{x}})$ into one isolated CZ gate and the remainder. $U_{\rm CNOT}^{\dag} P_{\rm CZ}(S^{\bm{x}}|_{\rm condition}) U_{\rm CNOT}$ represents the results of the decoupling. The PL moment column describes the relation between the $N$-qubit PL moment with fixed $x$ values and the PL moment of a state with a reduced number of qubits. Here, $m^{(N)}_{\alpha,(y_0, y_1, y_2, y_3)}$ is an abbreviation of $m_{\alpha,(x_0=y_0,x_1=y_1,x_2=y_2,x_3=y_3)}^{(N)}$. $\beta = 2^{(1-\alpha)2}$ and $y_0$, $y_1$, $y_2$, and $y_3$ are GF(2) variables. $E'_3$ and $E''_3$ are defined in the main text.
}
\begin{tabular}{c|cc}
Case & Condition & $U_{\rm CNOT}$
\\ \hline
(i) & $x_2=1$ & $CX_{2, 1}^{x_1} CX_{2, 0}^{x_0+x_3} CX_{3,0}$
\\
(ii) & $x_1=1,x_2=0$ & $CX_{1,0}^{x_0 + x_3} CX_{3,0}^{1+x_4} CX_{4,0}^{x_3}$
\\
(iii) & $x_1=x_2=0,x_0\oplus x_3=1$ & $CX_{3,1}^{x_4} CX_{4,1}^{x_3}$
\\
(iv) & $x_1=x_2=0, x_0 \oplus x_3=0$ & 1
\end{tabular}
\\[\baselineskip]
\begin{tabular}{c|cc}
Case &  $U_{\rm CNOT}^{\dag} P_{\rm CZ}(S^{\bm{x}}|_{\rm condition}) U_{\rm CNOT}$ & PL moment
\\ \hline
(i) & $CZ_{0,1} Z_2^{x_0x_1 + x_1x_3} \times \prod_{i, j,k: (i,j,k) \in E''_3} CZ_{j,k}^{x_i}|_{x_2=1}$ & $m^{(N)}_{\alpha,(y_0,y_1,1,y_3)} = \beta m^{(N-2)}_{\alpha,(1,y_3)}$
\\
(ii) & $CZ_{0,2} \times \prod_{i, j,k: (i,j,k) \in E'_3} CZ_{j,k}^{x_i}$ & $m^{(N)}_{\alpha,(y_0,1,0,y_3)} = \beta m^{(N-3)}_{\alpha,(y_3)}$
\\
(iii) & $CZ_{1,2} \times \prod_{i, j,k: (i,j,k) \in E'_3} CZ_{j,k}^{x_i}$ & $m^{(N)}_{\alpha,(y_3 \oplus 1 ,0,0,y_3)} = \beta m^{(N-3)}_{\alpha,(y_3)}$
\\
(iv) & $CZ_{2,3}^{x_4} CZ_{2,4}^{x_3} \times \prod_{i, j,k: (i,j,k) \in E'_3} CZ_{j,k}^{x_i}$ & $m^{(N)}_{\alpha,(y_3,0,0.y_3)} = m^{(N-2)}_{\alpha,(0,y_3)}$
\end{tabular}
\label{table:list_of_decoupling}
\end{table*}

By considering other possible decouplings, we can obtain similar relations. We have four patterns of decoupling: We can extract a CZ gate on (i) $(0,1)$, (ii) $(0,2)$, (iii) $(1,2)$ qubits, or (iv) we cannot extract a CZ gate on any pair of $(0,1,2)$ qubits. This classification corresponds to the pattern of a CZ gate originating from the hyperedge $(0,1,2)$. If more than one case applies, we prioritize the case with the lower label. We summarize the results of decoupling for each case in Table~\ref{table:list_of_decoupling}.

Now, we are ready to get the recursion relations of the PL moment. For example, we consider $m_{\alpha,(x_0=0,x_1=0)}^{(N)}$. It can be decomposed as
\begin{align}
m_{\alpha,(x_0=0,x_1=0)}^{(N)} = \frac{1}{4} \sum_{y_2,y_3 \in \{0,1\}^2} m_{\alpha,(x_0=0,x_1=0,x_2=y_2,x_3=y_3)}^{(N)}.
\end{align}
Using the results in Table~\ref{table:list_of_decoupling}, we obtain
\begin{align}
m_{\alpha, (0,0)}^{(N)} = \frac{1}{4} m_{\alpha, (0,0)}^{(N-2)} + \frac{\beta}{2} m^{(N-2)}_{\alpha, (1)} + \frac{\beta}{4} m_{\alpha, (1)}^{(N-3)},
\label{eq:recursion_relations_1}
\end{align}
where we use abbreviations
\begin{align}
m^{(N)}_{\alpha,(y_0)} &= m^{(N)}_{\alpha,(x_0=y_0)},
\\
m^{(N)}_{\alpha,(y_0, y_1)} &= m^{(N)}_{\alpha,(x_0=y_0,x_1=y_1)}.
\end{align}
Similarly, we obtain the recursion relations of $m_{\alpha, (1)}^{(N)}$ and $m_{\alpha,(0,1)}^{(N)}$ as
\begin{align}
m_{\alpha, (1)}^{(N)} &= \frac{\beta}{2} m_{\alpha,(1)}^{(N-2)} + \frac{1}{8} m_{\alpha,(0,1)}^{(N-2)} + \frac{\beta}{4}m_{\alpha,(0)}^{(N-3)} + \frac{\beta}{8} m_{\alpha, (1)}^{(N-3)},
\label{eq:recursion_relations_2}
\\
m_{\alpha,(0,1)}^{(N)} &= \frac{\beta}{2} m_{\alpha, (1)}^{(N-2)} + \frac{\beta}{2} m_{\alpha}^{(N-3)}.
\label{eq:recursion_relations_3}
\end{align}

We can express the recursion relations~\eqref{eq:recursion_relations_1}, \eqref{eq:recursion_relations_2}, and \eqref{eq:recursion_relations_3} in a matrix form. We define matrices $A$, $B$, and $C$ as
\begin{align} \notag
A=\begin{pmatrix}
\frac{1}{4} & 0 & \frac{\beta}{2} \\
0 & 0 & \frac{\beta}{2} \\
0 & \frac{1}{8} & \frac{\beta}{2}
\end{pmatrix}
,~B=
\begin{pmatrix}
0 & \frac{\beta}{4} \\
\frac{\beta}{4} & \frac{\beta}{4} \\
\frac{\beta}{4} & \frac{\beta}{8}
\end{pmatrix}
,~C=
\begin{pmatrix}
1/2 & 1/2 & 0
\\
0 & 0 & 1
\end{pmatrix}.
\end{align}
The recursion relations~\eqref{eq:recursion_relations_1}, \eqref{eq:recursion_relations_2}, and \eqref{eq:recursion_relations_3} can be combined into a single matrix equation:
\begin{align} \notag
\begin{pmatrix}
m_{\alpha, (0,0)}^{(N)} \\ m_{\alpha, (0,1)}^{(N)} \\ m_{\alpha, (1)}^{(N)}
\end{pmatrix}
=
A
\begin{pmatrix}
m_{\alpha, (0,0)}^{(N-2)} \\ m_{\alpha, (0,1)}^{(N-2)} \\ m_{\alpha, (1)}^{(N-2)}
\end{pmatrix}
+
B
\begin{pmatrix}
m_{\alpha, (0)}^{(N-3)}
\\
m_{\alpha, (1)}^{(N-3)}
\end{pmatrix}.
\end{align}
Furthermore, we define the vector $\mathcal{M}_{N}$ and the matrix $\mathcal{A}$ as
\begin{align} \notag
\mathcal{M}_{N} = \begin{pmatrix}
m_{\alpha, (0,0)}^{(N)} \\ m_{\alpha, (0,1)}^{(N)} \\ m_{\alpha, (1)}^{(N)} \\ m_{\alpha, (0,0)}^{(N-1)} \\ m_{\alpha, (0,1)}^{(N-1)} \\ m_{\alpha, (1)}^{(N-1)} \\ m_{\alpha, (0)}^{(N-2)} \\ m_{\alpha,(1)}^{(N-2)}
\end{pmatrix},
~
\mathcal{A} = \begin{pmatrix}
0_{3 \times 3} & A & B
\\
I_{3} & 0_{3 \times 3} & 0_{3 \times 2}
\\
0_{2 \times 3} & C & 0_{2 \times 2}
\end{pmatrix},
\end{align}
where $0_{n \times m}$ is the $n$-by-$m$ zero matrix and $I_{n}$ is the identity matrix with size $n$, we obtain the matrix recursion relation:
\begin{align}
\mathcal{M}_{N} = \mathcal{A} \mathcal{M}_{N-1}.
\end{align}

The initial conditions of the recursion relations~\eqref{eq:recursion_relations_1}, \eqref{eq:recursion_relations_2},  and \eqref{eq:recursion_relations_3} are obtained by brute force computations:
\begin{equation}
\begin{gathered}
m^{(3)}_{\alpha, (0,0)} = \frac{1+ \beta}{2}, m^{(3)}_{\alpha, (0,1)} = \beta, m^{(3)}_{\alpha, (1)} = \beta,
\\
m^{(4)}_{\alpha, (0,0)} = \frac{1+ 3 \beta}{4}, m^{(4)}_{\alpha, (0,1)} = \beta, m^{(4)}_{\alpha,(1)} = \frac{1+7\beta}{8}.
\end{gathered}
\end{equation}
Solving the recursion relations with the initial conditions and using
\begin{align}
m_{\alpha}^{(N)} = \frac{m_{\alpha, (0,0)}^{(N)} + m_{\alpha, (0,1)}^{(N)}}{4} + \frac{m_{\alpha, (1)}^{(N)}}{2},
\end{align}
we can exactly evaluate $m_{\alpha}^{N}$. For example, $\alpha=2$ SRE for large $N$ is given by
\begin{align}
S_{\alpha=2} \approx 0.6637N - 0.9125~(N \to \infty).
\end{align}

The analysis presented here can be applied to other one-dimensional hypergraph states. The limitation is similar to that of the transfer matrix method: it requires regularity and low connectivity. In the above example, after the CZ gate decoupling, the generalized stabilizer structure of the remaining part is similar to that before the decoupling. This is due to the regularity of the hypergraph structure. As for the low connectivity, we have examined $2^4$ patterns corresponding to $x_0,x_1,x_2$, and $x_3$. If there are more hyperedges containing either $(0,1)$, $(0,2)$, or $(1,2)$ qubit pairs, we would need to investigate more patterns. For these reasons, we may have to rely on numerical computations for more complicated hypergraph states.

\subsection{Numerical evaluation of SRE of triangulated lattice hypergraph states}
\label{sec:result_SRE_2d}

\begin{figure}[t]
\centering
\includegraphics[width=\columnwidth]{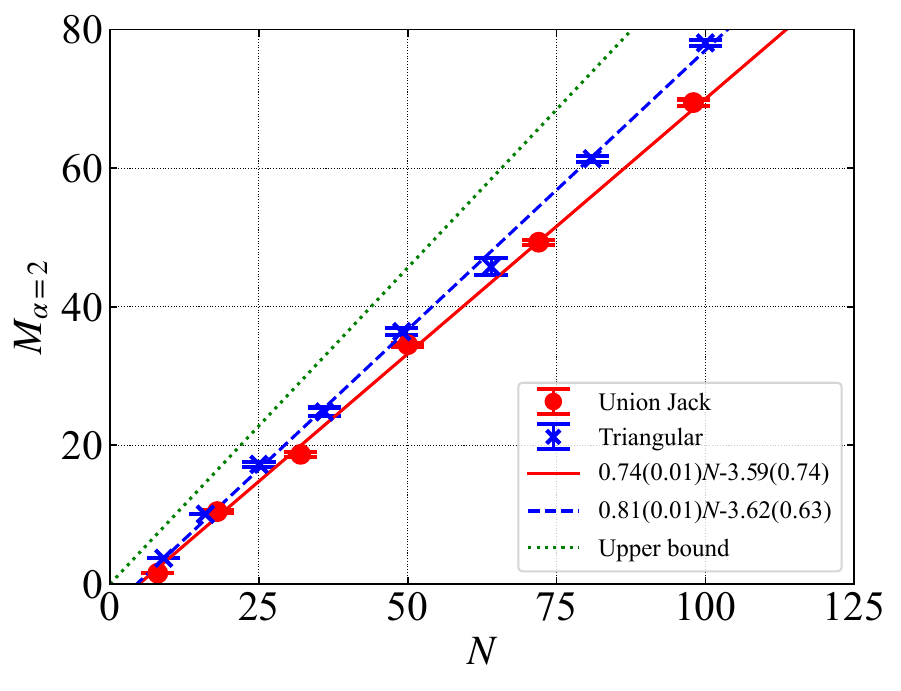}
\caption{
Evaluated $\alpha=2$ SREs of Union Jack and triangular lattice hypergraph states. Dots and crosses with error bars represent SREs of Union Jack and triangular lattice hypergraph states, respectively. Solid and dashed lines show linear fits. The dotted line is the upper bounds of $\alpha=2$ SRE of these states, Eq.~\eqref{eq:upper_bound_present_work} with $\bar{h} = 2$.
}
\label{fig:2d_hypergraph}
\end{figure}

Next, we evaluate the SREs of more complicated hypergraph states. We take the Union Jack lattice and the triangular lattice hypergraph states, shown in Fig.~\ref{fig:2d_hypergraph}(a) and (b), respectively. In the context of MBQC, the Union Jack lattice hypergraph state is Pauli-universal, but the triangular lattice one is considered not to be~\cite{Liu2022}. Thus, it is interesting to quantify SREs of these states.

We briefly explain how to numerically evaluate the SRE, or the PL moment, in a Monte Carlo method. The elementary Monte Carlo method gives the approximation of the PL moment as, for $\alpha = 2$, 
\begin{align}
m_{\alpha=2} \approx \frac{1}{M} \sum_{i=0}^{M-1} 2^{-2h(\bm{x}_i)},
\end{align}
where $M$ is the number of samples and $\bm{x}_i$ is the $i$th sample from the uniform distribution of $\bm{x}$. The choice of a more efficient distribution makes the calculation faster and more accurate, but it suffices for the present study. In the following calculation, we take $M = 2^8$.

Figure \ref{fig:2d_hypergraph} shows the calculated SREs of these states. Overall, the triangular lattice hypergraph state has a larger SRE than the Union Jack lattice one. This is consistent with the magic estimation in previous work~\cite{Liu2022}, in which it was found that the upper bound of the relative entropy of magic of the former state is larger than that of the latter. It was also proved that a state possessing too much nonstabilizerness is not Pauli-universal. These imply that a state having moderate nonstabilizerness is necessary for the Pauli universality. Can we distinguish the state having the Pauli universality using SREs? We do not have a clear answer to this question at this stage. We leave it as a future problem.

\section{Conclusion}
\label{sec:summary}

We have investigated the stabilizer R\'enyi entropy (SRE) of 3-uniform hypergraph states. We have shown that the SRE of 3-uniform hypergraph states can be expressed in terms of the rank of a symmetric GF(2) matrix $C(\bm{x}) + C(\bm{x})^T$, where $C(\bm{x})$ is an upper triangular GF(2) matrix determined by the structure of a hypergraph and a bit-string $\bm{x}$ corresponding to the $X$ operators in a Pauli string in the expression for the Pauli-Liouville moment. This expression enables us to exactly evaluate SREs of one-dimensional hypergraph states, as well as to make efficient evaluations of SRE. The origin of the efficiency would be attributed to the Aaronson-Gottesmann algorithm of evaluating the inner product of two stabilizer states. Utilizing this efficiency, we have numerically evaluated SREs for Union Jack and triangular lattice hypergraph states with numerous qubits.

As a future problem, it would be worthwhile to explore the relation between SREs and the physical properties of 3-uniform hypergraph states. In particular, it is intriguing to investigate whether the characteristics of SREs are related to the Pauli universality of measurement-based quantum computation (MBQC). Another interesting future direction is to extend our method to other states, such as $n$-uniform~$(n>3)$ hypergraph states and weighted graph states~\cite{Dur2005,Hein2006}. Since some of the weighted graph states are known to exhibit the Pauli universality~\cite{Kissinger2019}, this extension may give us useful insight into the relation between nonstabilizerness and MBQC. Since hypergraph states are employed in various contexts, including demonstrations of quantum advantage, MBQC, and the study of topological phases, our results would contribute to a deeper understanding of the role of nonstabilizerness in these areas.

\acknowledgments
The authors acknowledge Yuichiro Matsuzaki and Yuki Takeuchi for Insightful comments.
DK is supported by JSPS KAKENHI Grant Number JP25K17319.
ST is supported by JSPS KAKENHI Grant Number 25K07158.

\appendix
\section{Review of the theory of quadratic forms over GF(2)}
\label{appendix_A}

In this appendix, we review the theory of quadratic forms over GF(2) following section 6 in Ref.~\cite{Lidl_Niederreiter_1996}. The purpose of this appendix is to show Eq.~\eqref{eq:GF2_standard_form}. Even though Ref.~\cite{Lidl_Niederreiter_1996} deals with the quadratic form over general finite fields, we focus on the GF(2) case only and slightly modify the discussion to the GF(2) case.

In this appendix, we denote a quadratic form by $f(a_0,a_1,\dots,a_{N-1}) = \sum_{0 \leq i \leq j \leq N-1} C_{i,j} a_i a_j$. Note that $f$ includes the diagonal elements $a_i^2 = a_i$. We call a quadratic form $f$ in $N$ variables nondegenerate if $f$ cannot be transformed into a quadratic form in fewer than $N$ variables. Otherwise, we call it degenerate. To show Eq.~\eqref{eq:GF2_standard_form}, it is sufficient to discuss the nondegenerate case only. Before proceeding to the proof of Eq.~\eqref{eq:GF2_standard_form}, we prove the following lemma.

\begin{lemma}[6.29 in Ref.~\cite{Lidl_Niederreiter_1996}]
A nondegenerate quadratic form in $N(\geq 3)$ variables $f$ can reduce to $f(a_0,a_1,\dots,a_{N-1}) = a_0 a_1 + g(a_2, \dots, a_{N-1})$, where $g(a_2, \dots, a_{N-1})$ is a nondegenerate quadratic form in $N-2$ variables.
\end{lemma}

\begin{proof}
First, we show that $f$ can always be transformed to a quadratic form with $C_{0,0} = 0$ by a linear transformation. If some $C_{i,i}$ is zero, we obtain $C_{0,0} = 0$ by renaming variables. Thus, we consider the case where all $C_{i,i}= 1$. Next, if we had $C_{i,j} = 0$ for all $i < j$, then
\begin{align}
f = a_0^2 + a_1^2 + \dots + a_{N-1}^2 = (a_0 + a_1 + \dots + a_{N-1})^2,
\end{align}
which can be transformed to a quadratic form in one variable by $a_0 = a_0 + a_1 + \dots + a_{N-1}$, i.e., $f$ is degenerate. Therefore, we assume $C_{1,2} = 1$ by renaming variables suitably. We separate the terms of $f$ involving $a_{1}$ as
\begin{align} \notag
f &= a_1^2 + a_1 (C_{0,1} a_0 + a_{2} + \dots + C_{1,N-1} a_{N-1})
\\
&+ a_0^2 + \sum_{2 \leq j \leq N-1} C_{0,j} a_0 a_j
+ \sum_{2\leq i \leq j \leq N-1} C_{i,j} a_{i} a_{j},
\end{align}
and carrying out the linear transformation
\begin{equation}
a_{2} = C_{0,1} a'_0 + a'_2 + \sum_{j=3}^{N-1} C_{1,j} a'_{j},~a_{i} = a'_{i}~\mathrm{for}~i\neq 2.
\end{equation}
After the substitution, let $C'_{i,j}$ be the coefficient of $a'_i a'_j$;
\begin{align} \notag
f &= (a'_1)^2 + a'_1 a'_2 + C'_{0,0} (a'_0)^2 + \sum_{2 \leq j \leq N-1} C'_{0,j} a'_0 a'_j
\\
&\quad + \sum_{2\leq i \leq j \leq N-1} C'_{i,j} a'_{i} a'_{j}.
\end{align}
Further, we apply the linear transformation
\begin{equation}
a'_1 = C'_{0,0} a''_0 + a''_1,~a'_{i} = a''_{i}~\mathrm{for}~i\neq 1,
\end{equation}
getting
\begin{align} \notag
f
&= (a''_1)^2 + C'_{0,0} a''_0 a''_2 + a''_1 a''_2
\\
&\quad + \sum_{2 \leq j \leq N-1} C'_{0,j} a''_0 a''_j
+ \sum_{2\leq i \leq j \leq N-1} C'_{i,j} a''_{i} a''_{j}.
\end{align}
This is the desired form, i.e., a quadratic form with $C_{0,0} = 0$.

Now let $f$ be a quadratic form with $C_{0,0} = 0$:
\begin{align}
f = \sum_{1 \leq j \leq N-1} C_{0,j} a_0 a_j + \sum_{1 \leq i \leq j \leq N-1} C_{i,j} a_i a_j.
\end{align}
Since $f$ is nondegenerate, not all $C_{0,j}$ can be zero, so that we assume $C_{0,1} = 1$. We apply the linear transformation
\begin{equation}
a_1 =a'_1 + C_{0,2} a'_2 + \dots + C_{0,N-1} a'_{N-1},~a_{i} = a'_{i}~\mathrm{for}~i\neq 1.
\end{equation}
We write $f$ in the form $f = a'_0 a'_1 + \sum_{1 \leq i \leq j \leq N-1} C'_{i,j} a'_i a'_j$ after the transformation. The linear transformation
\begin{equation}
a'_0 =a''_0 + C'_{1,1} a''_1 + \dots + C'_{1,N-1} a''_{N-1},~a'_{i} = a''_{i}~\mathrm{for}~i\neq 0,
\end{equation}
yields
\begin{align}
f
=
a''_0 a''_1 +  \sum_{2 \leq i \leq j \leq N-1} C'_{i,j} a''_i a''_j.
\end{align}
Writing $g = \sum_{2 \leq i \leq j \leq N-1} C'_{i,j} a''_i a''_j$, we prove the statement of the lemma. Because $f$ is nondegenerate, $g$ is also nondegenerate.
\end{proof}

This lemma leads to the following theorem, which is the same statement as Eq.~\eqref{eq:GF2_standard_form}.

\begin{theorem}[6.30 in Ref.~\cite{Lidl_Niederreiter_1996}]
Let $f$ be a nondegenerate quadratic form having $N$ GF(2) variables. If $N$ is odd, $f$ is equivalent to
\begin{align}
f = a_0a_1 + a_2a_3 + \dots + a_{N-3} a_{N-2} + a_{N-1},
\label{eq:theorem_odd}
\end{align}
and, if $N$ is even, $f$ is equivalent to
\begin{align} \notag
f &= a_0 a_1 + a_2 a_3 + \dots + a_{N-4} a_{N-3} + a_{N-2} a_{N-1}
\\
&\quad + \eta( a_{N-2} + a_{N-1}),
\label{eq:theoreom_even}
\end{align}
under a GF(2) linear transformation. Here $\eta = 0$ or 1.
\end{theorem}

\begin{proof}
If $N$ is odd, inductively using Lemma, $f$ is equivalent to $a_0a_1 + a_2a_3 + \dots + a_{N-3} a_{N-2} + a_{N-1}^2$ and $a_{N-1}^2 = a_{N-1}$, we obtain Eq.\,\eqref{eq:theorem_odd}. If $N$ is even, using induction, $f$ is equivalent to $a_0 a_1 + a_2 a_3 + \dots + a_{N-4} a_{N-3} + b a_{N-2}^2 + c a_{N-2} a_{N-1} + d a_{N-1}^2$.
We examine the last three terms, $b a_{N-2}^2 + c a_{N-2} a_{N-1} + d a_{N-1}^2$. For $c=0$, we have $b a_{N-2}^2 + d a_{N-1}^2 = ba_{N-2} + d a_{N-1}$. Replacing $a_{N-2} = ba_{N-2} + d a_{N-1}$ implies that $f$ is degenerate, which contradicts the assumption. Thus, we can put $c=1$. When $b=d=0$, we trivially obtain Eq.\,\eqref{eq:theoreom_even} with $\eta = 0$. If $b = 0$ and $d=1$, $a_{N-2} a_{N-1} + a_{N-1}^2 = (a_{N-2} + a_{N-1}) a_{N-1}$. By replacing $a_{N-2} = a_{N-2} + a_{N-1}$, we obtain Eq,\,\eqref{eq:theoreom_even} with $\eta = 0$. We can reach the same conclusion when $b=1$ and $d=0$. For $b=c=d=1$, we cannot further simplify the equation by the linear transformation, so that we obtain Eq.\,\eqref{eq:theoreom_even} with $\eta = 1$.
\end{proof}

\section{Relation between the rank of $C+C^T$ and the non-degenerate subspace dimension of quadratic form $Q$}
\label{appendix_B}

In this appendix, we show the relation between the rank $2h$ of $C+C^T$ and the non-degenerate subspace dimension $r$ of the quadratic form $Q = \bm{a}^T C \bm{a}$, i.e., $r = 2h+1 (2h)$ when $r$ is odd (even).

From the discussion of Appendix \ref{appendix_A}, there exists a GF(2) matrix $P$ such that the linear transformation $\bm{a} = P \bm{a}'$ makes $C$ into
\begin{align}
\tilde{C} =
\begin{cases}
\displaystyle \bigoplus_{k=1}^{(r-1)/2} \begin{pmatrix}
0 & 1 \\ 0 & 0
\end{pmatrix} \oplus I_{1} \oplus 0_{N-r} & (r: \text{odd}),
\\
\displaystyle \bigoplus_{k=1}^{r/2-1} \begin{pmatrix}
0 & 1 \\ 0 & 0
\end{pmatrix} \oplus 
\begin{pmatrix}
\eta & 1 \\ 0 & \eta
\end{pmatrix} \oplus 
0_{N-r} & (r: \text{even}),
\end{cases}
\end{align}
where $0_{n}$ is the $n$-by-$n$ zero matrix and $I_{n}$ is the identity matrix with size $n$. In other words, $\tilde{C} = P^T C P + S$, where $S$ is a symmetric matrix with zero diagonal entries satisfying $S_{i,j} = (P^T C P)_{i,j}$ for $i > j$. $S$ makes $\tilde{C}$ an upper triangular matrix. Using this representation, $\tilde{C}+\tilde{C}^T = P^T (C + C^T) P$, since $S+S^T=2S=0$ over GF(2). Thus, the same linear transformation $\bm{a} = P \bm{a}'$ transforms $C + C^T$ as
\begin{align}
\tilde{C} + \tilde{C}^T = \bigoplus_{k=1}^{h} \begin{pmatrix}
0 & 1 \\ 1 & 0 
\end{pmatrix} \oplus 0_{N-2h}.
\end{align}
This implies that the rank $2h$ of $C+C^{T}$ is the same as twice the number of quadratic pairs in the quadratic form $Q$, that is, $r = 2h + 1~(2h)$ when $r$ is odd (even).

\section{Detailed calculation of $h_k$ and $\Delta_k$}
\label{appendix:h_Delta_calc}

In this appendix, we show detailed calculations of $h_k$ and $\Delta_k$ of the Union Jack and triangular lattice hypergraph states in Sec.~\ref{sec:Result_upper_bound}.

We consider $L \times L$ Union Jack and triangular lattices with periodic boundary conditions. To properly define hypergraph states, $L \geq 2~(3)$ is necessary for the Union Jack (triangular) lattice. Otherwise, some hyperedges contain the same qubit twice, or the resulting state is trivial, $\ket{+}^{\otimes N}$.

First, we consider the Union Jack lattice hypergraph state for $L \geq 3$. Since the $L = 2$ case is special, we discuss it later. The unfilled and filled dots (qubits) in Fig.~\ref{fig:lattice_unionjack_triangular}(a) are, respectively, connected to 8 and 4 qubits by hyperedges. The non-zero parts of the matrix $C'_{k}$ corresponding to the filled and unfilled qubits, denoted by $C'_{k,{\rm unfilled}}$ and $C'_{k,{\rm filled}}$, are given by
\begin{gather}
C'_{k,{\rm unfilled}} = \begin{pmatrix}
0 & 1 & 0 & 0 & 0 & 0 & 0 & 1 \\
& 0 & 1 & 0 & 0 & 0 & 0 & 0 \\
&  & 0 & 1 & 0 & 0 & 0 & 0 \\
&	&	& 0 & 1 & 0 & 0 & 0 \\
&	&	&	& 0 & 1 & 0 & 0 \\
&	&	&	& & 0 & 1 & 0 \\
&	&	&	& &  & 0 & 1 \\
&	&	&	& &  &  & 0 \\
\end{pmatrix},
\\
C'_{k,{\rm filled}} = \begin{pmatrix}
0 & 1 & 1 & 0 \\
   & 0 & 0 & 1 \\
   &   & 0 & 1 \\
   &	&	& 0
\end{pmatrix}.
\end{gather}
The ranks of symmetrized versions of these GF(2) matrices are $2h_{k,{\rm unfilled}} = 6$ and $2h_{k,{\rm filled}} = 2$, respectively. Since $\Delta_k = \sum_{i<j} (C'_{k})_{i,j}$, $\Delta_{k,{\rm unfilled}}=8$ and $\Delta_{k,{\rm filled}} = 4$. Taking average, we obtain $\bar{h} = 2$ and $\bar{\Delta} = 6$ for $L \geq 3$.

For the $L=2$ Union Jack lattice case, unfilled qubits are connected to 6 qubits by hyperedges, whereas filled qubits are the same as the $L \geq 3$ case. Due to periodic boundary conditions, the unfilled qubits above and below a given unfilled qubit become identical, and likewise, the unfilled qubits to its left and right are identical. The non-zero parts of the matrix $C'_{k}$ corresponding to unfilled qubits, $C'_{k,{\rm unfilled},L=2}$, is given by
\begin{align}
C'_{k,{\rm unfilled},L=2} =
\begin{pmatrix}
0 & 1 & 1 & 0 & 1 & 1 \\
   & 0 & 0 & 1 & 0 & 0 \\
   &   & 0 & 1 & 0 & 0  \\
   &	&	& 0 & 1 & 1  \\
   &	&	&	& 0 & 0 \\
      &	&	&	&    & 0 \\
\end{pmatrix}.
\end{align}
The corresponding symmetrized matrix is of rank $2h_{k,{\rm unfilled},L=2} = 2$. Consequently, we obtain $\bar{h} = 1$. Note that $\Delta_{k,{\rm unfilled},L=2}$ equals 8, which is the same as in the $L \geq 3$ case.

As for the triangular lattice hypergraph state, every qubit is connected to 6 qubits by hyperedges. The non-zero part of the matrix $C'_{k}$, expressed as $C'_{k,{\rm triangular}}$, has the same form for every qubit, which is given by
\begin{align}
C'_{k,{\rm triangular}}  = \begin{pmatrix}
0 & 1 & 0 & 0 & 0 & 1 \\
   & 0 & 1 & 0 & 0 & 0 \\
   &   & 0 & 1 & 0 & 0 \\
   &	&	& 0 & 1 & 0 \\
   &	&	&	& 0 & 1  \\
      &	&	&	&    & 0
\end{pmatrix}.
\end{align}
The symmetrized form of this matrix has the rank $2h_{k,{\rm trinangular}} = 2\bar{h} = 4$. In the same way as the Union Jack lattice case, we obtain $\Delta_{k,{\rm triangular}} = \bar{\Delta} = 6$.

\bibliography{references.bib}

\onecolumngrid

\end{document}